\documentclass{article}
\usepackage{iclr2027_conference,times}
\iclrfinalcopy
\usepackage[T1]{fontenc}
\usepackage[utf8]{inputenc}
\usepackage{amsmath,amsthm}
\usepackage{newtxmath}
\usepackage{bm}
\usepackage{graphicx}
\usepackage{booktabs,tabularx,array,multirow,makecell}
\usepackage{algorithm}
\usepackage{algpseudocode}
\usepackage{float}
\usepackage{enumitem}
\usepackage[section]{placeins}
\usepackage{xcolor}
\usepackage{microtype}
\usepackage{hyperref}
\usepackage{url}
\hypersetup{hidelinks}
\usepackage[table]{xcolor}

\newtheorem{proposition}{Proposition}
\newtheorem{definition}{Definition}
\newtheorem{remark}{Remark}

\title{LoRo-Mark: Provably Lossless And Robust Agent Watermarking}

\author{
    HaoYang Zou \\ University of Science and Technology of China
    \And
    Yao Wang \\ University of Science and Technology of China
    \And
    Jun Yao \\ University of Science and Technology of China
    \And
    Weiming Zhang \\ University of Science and Technology of China
    \And
    Han Fang \\ University of Science and Technology of China
}

\begin{document}
\raggedbottom

\maketitle

\begin{abstract}
As large language model (LLM) agents are increasingly deployed as commercial services, protecting their proprietary orchestration logic and tool-use policies has become an important concern. We consider a realistic infringement scenario, termed \textbf{agent repackaging}, in which an adversary integrates a protected agent into its own application through an API and presents it under its own service identity. The adversary may further modify parts of its execution process to obscure the original source. In such cases, the owner typically has access only to the repackaged service interface, making black-box ownership verification essential. Agent watermarking provides a natural way to embed ownership evidence into the agent’s behavior for later verification. Under this setting, we argue that an effective watermark should satisfy two key requirements: losslessness, requiring watermarking to preserve the agent’s original functionality, and robustness, requiring ownership evidence to remain recoverable after partial modification of the execution process. Existing methods, however, typically embed ownership signals into behavior selection or execution trajectories, intervene in normal agent decisions and provide limited robustness to behavior modification. To address these limitations, we propose \textbf{LoRo-Mark}, a provably \textbf{lo}ssless and \textbf{ro}bust agent water\textbf{mark}ing mechanism. For losslessness, LoRo-Mark isolates watermarking into a cryptographically authenticated forensic branch that remains inactive during normal execution and can only be activated by owner-authorized requests. By reducing unauthorized branch activation to standard MAC security, LoRo-Mark provides a formal cryptographic guarantee of performance preservation. For robustness, it redundantly distributes ownership information across forensic behavior sequences, enabling reliable recovery under partial behavior substitution and sequence truncation. Experiments across multiple LLM agents demonstrate zero degradation on normal tasks and reliable ownership verification under sequence modifications.
\end{abstract}

\section{Introduction}

Large language model (LLM) agents are increasingly deployed in enterprise automation, software engineering, and other real-world applications~\citep{yang2024sweagent,zhou2023webarena,bran2023chemcrow,wang2023voyager}. Beyond the underlying foundation model, their practical value often lies in system-level capabilities such as task decomposition, tool-use strategies, domain-specific prompting, and the handling of intermediate observations, which together determine how an agent plans and executes complex multi-step tasks~\citep{xi2023llmagents,yao2023react,schick2023toolformer,park2023generative,liu2023agentbench,shinn2023reflexion,yao2023tree,wei2022chain,nakano2021webgpt,qin2023toolllm,patil2023gorilla,karpas2022mrkl}. Developing such orchestration logic typically requires substantial engineering effort and accumulated operational expertise, making it an increasingly important form of intellectual property. Protecting this system-level IP has therefore become a growing concern as agent-based services move toward large-scale deployment~\citep{wang2026agentwm}.
\begin{figure}[t]
  \centering
  \includegraphics[width=\linewidth]
  {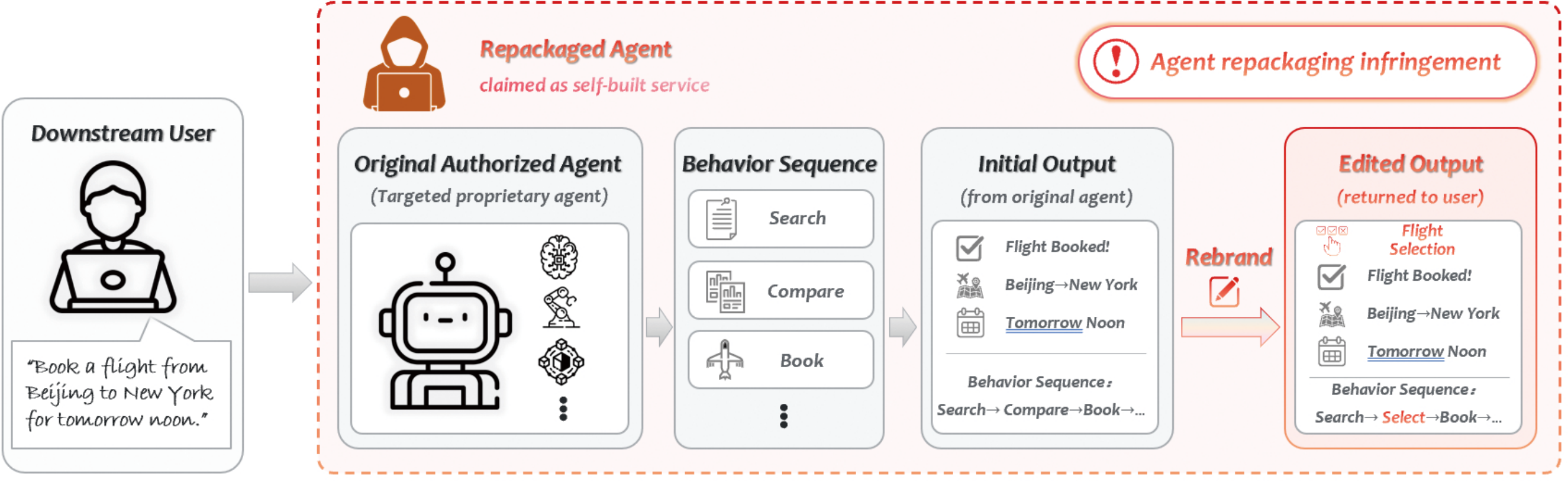}
  \caption{Scenario of repackaging infringement: A downstream user sends a request to the repackaged agent to book a specific flight. The repackaged agent completes the booking task via a sequence of tool calls and returns an edited output modified from the initial output to obscure the original source. }
  \label{fig:threat-model}
\end{figure}

Among the potential threats to agent IP, we focus on a realistic yet under-explored infringement scenario, which we term agent repackaging. As illustrated in Figure~\ref{fig:threat-model}, an adversary may integrate a protected agent into its own application through an API and present it under its own service identity. To obscure the original source, the adversary can alter the host application’s front-end presentation, including the displayed agent thinking process, behavior sequence, action trajectory, and observation results. Such repackaging allows the adversary to exploit the capabilities of a non-self-developed agent while presenting them as its own. More importantly, the original owner typically has no access to the internal implementation of the suspicious service and can only interact with its exposed interface, making black-box ownership verification particularly important. Agent watermarking therefore provides a natural means of embedding ownership evidence into agent behavior for subsequent verification~\citep{uchida2017embedding,adi2018turning,rouhani2018deepsigns,chen2018deepmarks,chen2019blackmarks,kirchenbauer2023watermark,christ2024undetectable,kuditipudi2024robust,dathathri2024scalable,huang2025agentguide,huang2026agentmark,wang2026agentwm}.

Under agent repackaging, an effective watermark should satisfy two key requirements: losslessness and robustness~\citep{christ2024undetectable,kuditipudi2024robust}. First, watermarking is introduced for ownership protection and should therefore preserve the agent’s original functionality during normal use. Second, a repackaging adversary may partially modify the agent’s execution process to obscure its source, requiring ownership evidence to remain recoverable after such modifications. Existing agent watermarks, however, typically encode ownership signals through specific behavior selections or execution trajectories~\citep{huang2025agentguide,huang2026agentmark,wang2026agentwm,an2026seqwm}. Such designs intervene in the agent’s decision process, making strict functionality preservation difficult, while the resulting ownership evidence remains vulnerable to downstream behavior modification.

These limitations motivate us to rethink how watermarking should interact with agent execution. For provable losslessness, the key is to maintain an untouched normal execution path while making unauthorized activation of the watermark branch cryptographically negligible. Rather than minimizing watermark-induced perturbations, we preserve the original execution path for ordinary requests and activate watermarking only through a MAC-authenticated forensic branch~\citep{bellare1996keying,goldwasser2022backdoors,nist2008hmac}. For robustness, we redundantly distribute ownership information across the forensic behavior sequence, preventing the evidence from depending on any single behavior and enabling its recovery after partial behavior modification or truncation~\citep{lin2004error}.

Building on these principles, we propose \textbf{LoRo-Mark}, a provably \textbf{lo}ssless and \textbf{ro}bust agent water\textbf{mark}ing mechanism for black-box ownership verification. LoRo-Mark uses cryptographic authentication to control access to the forensic branch and redundantly embeds copyright information across the resulting behavior sequence~\citep{bellare1996keying,nist2008hmac}. During verification, the owner uses the secret key to reconstruct the pseudorandom behavior mappings, translate the observed behaviors back into watermark bits, and recover the copyright message through redundancy-aware decoding~\citep{goldreich1986construct,lin2004error}. We formally reduce unauthorized branch activation to MAC-EUF-CMA security, establishing provable performance preservation for ordinary requests, and derive recovery guarantees under behavior substitution and sequence truncation. 

Our main contributions are summarized as follows:
\begin{itemize}[leftmargin=1.5em]

\item We focus on \textbf{agent repackaging}, a realistic yet under-explored infringement scenario in which a protected agent can be proxied, redeployed, and partially modified behind a black-box service. We identify \textbf{losslessness} and \textbf{robustness} as two key requirements for agent watermarking under this setting: preserving the agent's original functionality while retaining ownership evidence after repackaging.

\item We propose \textbf{LoRo-Mark}, a provably lossless and robust agent watermarking mechanism. LoRo-Mark maintains an untouched normal execution path through a MAC-authenticated forensic branch, making unauthorized activation cryptographically negligible, while redundantly encoding ownership information across forensic behavior sequences to enable reliable recovery under partial behavior substitution and sequence truncation.

\item We provide rigorous theoretical and empirical validation of both properties.  Extensive experiments across diverse agents, models, and tasks demonstrate zero degradation on normal execution and reliable ownership verification under substantial sequence modifications.

\end{itemize}

\section{Related Work}

\subsection{Agent Watermarking}
Existing agent watermarking approaches embed ownership signals across three principal loci---planning behaviors, execution trajectories, and memory-state evolutions. At the behavior layer, \textit{AgentGuide}~\citep{huang2025agentguide} targets high-level planning decisions rather than surface tokens~\citep{kirchenbauer2023watermark,christ2024undetectable,kuditipudi2024robust,dathathri2024scalable}. It biases the behavior distribution over a predefined set $\mathcal{B}$ via a pseudorandom green-list~\citep{kirchenbauer2023watermark} and uses a $z$-statistic aggregated over multiple rounds for single-bit presence detection. \textit{AgentMark}~\citep{huang2026agentmark} instead encodes multi-bit identifiers through \emph{distribution-preserving conditional sampling}, which preserves the marginal behavior distribution while remaining compatible with black-box agent APIs. At the trajectory layer, \textit{SeqWM}~\citep{an2026seqwm} conditions each embedding step on history-dependent state transitions; \textit{ACTHOOK}~\citep{meng2026acthook} injects auxiliary decoy actions activated by predefined trigger keys; and \textit{AGENTWM}~\citep{wang2026agentwm} discovers pools of \emph{semantically equivalent tool-execution paths} and biases the agent toward user-specific subsets, with verification based on Jensen--Shannon divergence. At the memory layer, \textit{MEMMARK}~\citep{zhang2026memmark} embeds signals into latent memory-state transitions, enabling provenance attribution from the final memory snapshot even when intermediate trajectories are unavailable.  

\subsection{Cryptographically Undetectable Backdoors}
Shafi Goldwasser and his coworkers demonstrated that a malicious trainer can plant undetectable backdoors in classifiers: no computationally bounded observer, whether black‑box or white‑box, can tell backdoored models apart from clean ones without the secret key~\citep{goldwasser2022backdoors}. Their design runs a digital‑signature verification circuit parallel to the base model, activating the backdoor solely for valid‑signature inputs while leaving others intact, with white‑box undetectability grounded in Sparse PCA hardness. The core insight is that a secret key provides cryptographically secure control over model behavior — key holders can reliably trigger target outputs, whereas observers without keys detect nothing. We draw inspiration from this key‑driven behavioral asymmetry. Nevertheless, agent scenarios bring unique hurdles including multi‑step reasoning, tool calls, and emergent plan‑level semantics, so we cannot directly reuse classifier‑level constructions and must build an agent‑tailored watermarking protocol instead.

\section{Methodology}
\subsection{Threat Model}
\textbf{Adversary Capabilities.}
We model a protected agent as $A=(M,\Pi,\mathcal{T})$, where $M$ denotes the underlying LLM, $\Pi$ the orchestration policy, and $\mathcal{T}$ the available tool set. At step $t$, the agent selects a high-level behavior $b_t\in\mathcal{B}_t$ from a candidate behavior set $\mathcal{B}_t$ and instantiates it into a concrete action $a_t$. A complete execution produces a behavior sequence $\tau_b=(b_1,\ldots,b_n)$ and the corresponding action trace $\tau_a=(a_1,\ldots,a_n)$~\citep{huang2025agentguide,huang2026agentmark,wang2026agentwm}. After obtaining access to the protected agent, the adversary may repackage it into a new service $A'=(M,\Pi,\mathcal{T})$, while preserving the underlying LLM, orchestration policy, available tools, and actual execution process. Instead, the adversary may modify the host application's front-end presentation, including the displayed agent thinking process, behavior sequence, action trajectory, and observation results, before exposing them to downstream users. We therefore view agent repackaging as a service-level transformation
\[
\mathcal{R}: A \rightarrow A',
\]where $\mathcal{R}$ leaves the execution flow, decision logic, tool set, and actual trajectory unchanged, but modifies their exposed presentation while still preserving the core utility and advertised functionality of the original agent. The adversary may know the watermarking mechanism and interact with the protected agent through public interfaces, but does not possess the owner's secret keys.

\textbf{Defender Goal.}
Given only black-box access to the repackaged agent $A'$, the defender aims to determine whether $A'$ still contains or invokes the protected agent $A$~\citep{huang2025agentguide,huang2026agentmark,wang2026agentwm}. 
The defender constructs an owner-authorized forensic request $x^{*}$ using the secret key $k$ and submits it to $A'$. 
The repackaged service returns an observable execution trace
$
\tau' = \bigl(b'_1,a'_1,\ldots,b'_m,a'_m\bigr),
$
which may already have been modified by the repackaging transformation $\mathcal{R}$.
The defender then applies a keyed extractor
$
\hat{c}=\mathsf{Ext}_k(\tau')
$
to recover the ownership message and accepts the ownership claim iff
$
\mathsf{Ver}_k(\hat{c},c)=1,
$
where $c$ is the claimed watermark payload. Accordingly, the verification goal is successful black-box watermark recovery from the repackaged execution trace, even when $\tau'$ has undergone partial modification or truncation.

\begin{figure}[t]
  \centering
  \includegraphics[
    width=0.9\linewidth,
    height=7.0cm,
    keepaspectratio
  ]{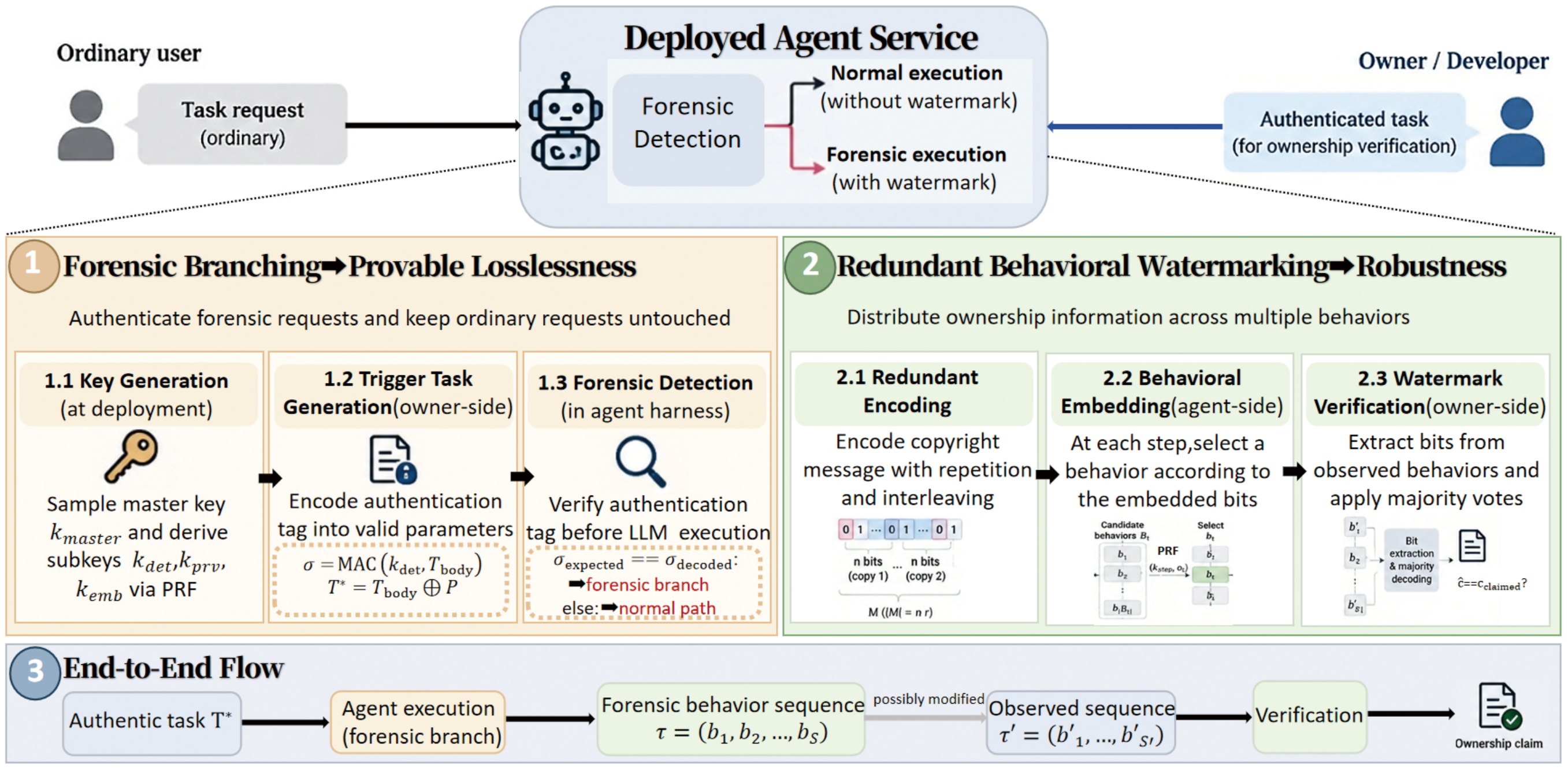}
  \caption{An overview of LoRo-Mark.}
  \label{fig:protocol-overview}
\end{figure}

\subsection{LoRo-Mark}
Figure~\ref{fig:protocol-overview} presents an overview of LoRo-Mark.
LoRo-Mark is designed around the two requirements identified above:
\emph{provable losslessness} and \emph{robustness}.
Accordingly, its design consists of two key components.
First, \emph{forensic branching} preserves
an untouched execution path for ordinary requests while making unauthorized
activation of the watermark branch cryptographically negligible, thereby
enabling provable losslessness.
Second, \emph{redundant behavioral watermarking}~\citep{lin2004error} distributes ownership
information across the forensic behavior sequence, allowing the watermark
to remain recoverable under partial sequence modification.




\subsubsection{Forensic Branching}
The key to provable losslessness is not merely to reduce watermark-induced
perturbations, but to ensure that ordinary requests remain on an execution
path untouched by watermarking, except with cryptographically negligible
probability.
To achieve this, we place all watermark functionality behind an authenticated
forensic branch.
An ordinary request follows the original agent execution path without watermark
intervention, whereas only a request carrying a valid owner-generated
authentication tag can activate the forensic branch.
We instantiate this design through three modules:
\textbf{Key Generation}, \textbf{Trigger Task Generation}, and \textbf{Forensic Detection}.

\paragraph{\textbf{Key Generation}.}
At deployment time, given a security parameter $\lambda$, the developer
samples a master key $k_{\mathrm{master}}$ and derives three subkeys
$\{
    k_{\det},k_{\mathrm{prv}},k_{\mathrm{emb}}\}
$
through domain-separated PRF evaluations~\citep{goldreich1986construct}.
The detection key $k_{\det}$ authenticates forensic requests,
$k_{\mathrm{emb}}$ determines how authentication information is encoded into
valid task parameters, and $k_{\mathrm{prv}}$ is reserved for keyed behavioral
watermarking in the forensic branch.
The owner retains these keys for forensic verification, while the deployed
agent stores the corresponding keys required for request detection and
watermark generation.


\paragraph{\textbf{Trigger Task Generation}.}
To initiate ownership verification, the owner constructs an authenticated
forensic request from an otherwise valid task $T_{\mathrm{body}}$.
Let $P$ denote a set of semantically valid task parameters, such as budget,
date, quantity, preference, or priority.
We first compute~\citep{bellare1996keying}
\[
    \sigma = \operatorname{MAC}(k_{\det},T_{\mathrm{body}}),
\]
and then use $k_{\mathrm{emb}}$ to encode $\sigma$ into the parameter values
$P$.
The resulting forensic request is
\[
    T^{*} = T_{\mathrm{body}} \oplus P,
\]
where $\oplus$ denotes assembling the task body with the encoded parameters.
Because $P$ occupies legitimate parameter slots and remains within their valid
domains, $T^{*}$ retains the syntax and semantics of a valid agent request.
When the available parameter capacity is smaller than the full MAC output,
the authentication tag can be truncated to $\ell$ bits; the resulting security--capacity trade-off is evaluated in Section~\ref{sec:5}.

\paragraph{\textbf{Forensic Detection}.}
Forensic detection (denoted as $\mathcal{F}$) is implemented in the agent harness and is performed before
LLM execution. Given an incoming request $T=(T_{\mathrm{body}},P)$, it computes
\begin{align}
  \sigma_{\mathrm{expected}}
    &= \operatorname{MAC}(k_{\det},T_{\mathrm{body}}),\\
  \sigma_{\mathrm{decoded}}
    &= \operatorname{\mathcal{D}}(P,k_{\mathrm{emb}}).
\end{align}
where $\mathcal{D}$ indicates the operation that decodes $\sigma_{\mathrm{decoded}}$ from the parameters $P$. If $\sigma_{\mathrm{expected}}=\sigma_{\mathrm{decoded}}$, the request is
routed to the forensic branch; otherwise, it proceeds along the original
execution path. Hence, activating the forensic branch without the owner's
keys requires producing a valid authentication tag. For any unauthorized
request $T$,
\begin{equation}
  \Pr[\mathcal{F}(T)=1]
  \leq \varepsilon_{\mathrm{MAC}}(\lambda),
\end{equation}
where $\varepsilon_{\mathrm{MAC}}(\lambda)$ is negligible under standard MAC
security~\citep{bellare1996keying}. Consequently, ordinary requests remain on the original execution
path except with cryptographically negligible probability, providing the
basis for LoRo-Mark's provable losslessness.

\subsubsection{Redundant Behavioral Watermarking}

LoRo-Mark embeds the ownership
information into the resulting behavior sequence.
The key design principle for robustness is to avoid relying on any individual
behavior as ownership evidence. Instead, we redundantly encode the copyright
message and distribute the encoded bits across multiple forensic behaviors.
Consequently, even if part of the behavior sequence is substituted or
truncated during repackaging, the surviving behaviors can still provide
sufficient evidence for watermark recovery.
This process consists of three stages: \textbf{Redundant Encoding},
\textbf{Behavioral Embedding}, and \textbf{Watermark Verification}.

\paragraph{\textbf{Redundant Encoding}.}
Let $c$ denote the copyright message and
$n=8|c|$ its bit length, with
$c=(c_0,c_1,\ldots,c_{n-1})$.
To distribute each copyright bit across different positions of the forensic
sequence, we encode $c$ with repetition redundancy $r$~\citep{lin2004error}.
Rather than placing the $r$ copies of each bit consecutively, we interleave
multiple copies of the entire copyright message:
\begin{equation}
    M =
    (c_0c_1\cdots c_{n-1})
    \,\|\,\cdots\,\|\,
    (c_0c_1\cdots c_{n-1}),
\end{equation}
where the message is repeated $r$ times and $|M|=nr$.
Hence, the $k$-th copy of copyright bit $c_i$ appears at position
$i+kn$ in $M$.
This interleaved layout spreads different copies of the same copyright bit
across distant positions of the behavior sequence, allowing them to survive
partial sequence corruption and truncation.

\paragraph{\textbf{Behavioral Embedding}.}
LoRo-Mark next maps the redundantly encoded bitstream $M$ into a sequence of
forensic behaviors.
Given the authentication tag $\sigma$ generated for the forensic request,
we derive a request-specific embedding key
\begin{equation}
    k_{\mathrm{step}}
    = \operatorname{MAC}(k_{\mathrm{prv}},\sigma).
\end{equation}
At execution step $t$, the agent observes $o_t$ and constructs an admissible
candidate behavior set $\mathcal{B}_t$.
Using $k_{\mathrm{step}}$ and $o_t$, LoRo-Mark generates a pseudorandom
permutation $\pi_t = \operatorname{PRF}(k_{\mathrm{step}},o_t,\mathcal{B}_t)$~\citep{goldreich1986construct} 
over the candidate behaviors.
The number of bits that can be embedded at this step is $q_t=\left\lfloor\log_2|\mathcal{B}_t|\right\rfloor$. LoRo-Mark reads the next $q_t$ bits from $M$, interprets them as an integer $j_t\in[0,2^{q_t}-1]$, and selects $b_t=\pi_t[j_t]$. Repeating this process across execution steps produces the forensic behavior
sequence
\[
    \tau=(b_1,b_2,\ldots,b_S),
\]
through which the redundant copyright bitstream is distributed across
multiple observable behaviors.

\paragraph{\textbf{Watermark Verification}.}
Given a behavior sequence $\tau'=(b'_1,\ldots,b'_{S'})$ observed from a
suspicious repackaged service, the owner reverses the embedding process using
the secret key.
For each surviving behavior $b'_t$, the verifier reconstructs the same
candidate set and pseudorandom ordering $\pi_t$, identifies the index of
$b'_t$ in $\pi_t$, and converts that index back into the corresponding
embedded bits.
Aggregating the extracted bits yields a recovered redundant bitstream
$\widehat{M}$.

The verifier then groups the surviving copies corresponding to each copyright
bit and performs majority decoding~\citep{lin2004error} to obtain
\[
    \hat{c}=(\hat{c}_0,\ldots,\hat{c}_{n-1}).
\]
The ownership claim is accepted iff
\(\hat{c}=c_{\mathrm{claimed}}\).

Because each copyright bit is represented at multiple positions,
as long as sufficient copies survive for decoding, the ownership message remains recoverable. We formalize these robustness guarantees under behavior
substitution and sequence truncation in Section~\ref{sec:supp-robustness}.

\section{Provable Losslessness}
\label{sec:losslessness}

The losslessness of LoRo-Mark follows from a simple observation:
the watermarked agent differs from the original agent only when the
forensic branch is activated.
Therefore, it suffices to show that an unauthorized request activates
this branch only with cryptographically negligible probability.
We first bound such unauthorized activation through the security of the
underlying MAC~\citep{bellare1996keying,nist2008hmac}, and then show that this bound directly implies
losslessness.

\begin{definition}[$(\varepsilon,t,q,\lambda)$-losslessness]
\label{def:losslessness}
Let $\mathrm{API}_{\mathrm{WM}}$ denote the LoRo-Mark-protected agent
and $\mathrm{API}_{\mathrm{plain}}$ the original agent.
LoRo-Mark is $(\varepsilon,t,q,\lambda)$-\emph{lossless} if, for every
PPT distinguisher $\mathcal{A}$ running in time at most $t$ and making
at most $q$ API queries without access to the owner's secret keys,
\begin{equation}
    \operatorname{Adv}_{\mathrm{loss}}(\mathcal{A})
    :=
    \left|
    \Pr[\mathcal{A}^{\mathrm{API}_{\mathrm{WM}}}\!\to 1]
    -
    \Pr[\mathcal{A}^{\mathrm{API}_{\mathrm{plain}}}\!\to 1]
    \right|
    \leq \varepsilon(\lambda).
\end{equation}
\end{definition}

\newtheorem{theorem}{Theorem}

\begin{theorem}[Provable Losslessness of LoRo-Mark]
\label{thm:losslessness}
Suppose the MAC scheme used by LoRo-Mark is
$(\varepsilon_{\mathrm{MAC}},t,q,\lambda)$-EUF-CMA secure.
Then LoRo-Mark is
$(\varepsilon_{\mathrm{MAC}},t,q,\lambda)$-lossless.
\end{theorem}

\begin{proof}
We prove the result in two steps.

\textbf{Step 1: Unauthorized forensic activation is negligible.}
Consider any request
$T=(T_{\mathrm{body}},P)$ that is not generated by the authorized
Trigger Task Generation process.
Forensic Detection accepts $T$ only if
\begin{equation}
    \operatorname{\mathcal{D}}(P,k_{\mathrm{emb}})
    =
    \operatorname{MAC}(k_{\det},T_{\mathrm{body}}).
\end{equation}
Hence, an adversary that successfully constructs an unauthorized request
accepted by Forensic Detection also produces a valid authentication tag
for $T_{\mathrm{body}}$.

Suppose an adversary $\mathcal{A}$ can cause such an activation with
probability $p$.
We construct a MAC forger $\mathcal{B}$ that samples
$k_{\mathrm{emb}}$ independently and runs $\mathcal{A}$.
Whenever $\mathcal{A}$ outputs an accepted request
$T=(T_{\mathrm{body}},P)$, $\mathcal{B}$ outputs
\begin{equation}
    \left(
    T_{\mathrm{body}},
    \operatorname{\mathcal{D}}(P,k_{\mathrm{emb}})
    \right)
\end{equation}
as a MAC forgery.
Therefore, by EUF-CMA security,
\begin{equation}
    \Pr[\mathsf{Bad}]
    :=
    \Pr\!\left[
        \exists\, T :
        \operatorname{\mathcal{F}}(T)=1
        \text{ for an unauthorized } T
    \right]
    \leq
    \varepsilon_{\mathrm{MAC}}(\lambda).
\end{equation}

\textbf{Step 2: No forensic activation implies identical normal execution.}
Conditioned on $\neg\mathsf{Bad}$, every query issued by
$\mathcal{A}$ follows the original execution path of the agent.
LoRo-Mark does not modify the underlying model, orchestration policy,
tool invocation, or output generation on this path.
Consequently,
$\mathrm{API}_{\mathrm{WM}}$ and
$\mathrm{API}_{\mathrm{plain}}$ induce identical output distributions
conditioned on $\neg\mathsf{Bad}$.

The two APIs can therefore differ only if $\mathsf{Bad}$ occurs, yielding
\begin{align}
    \operatorname{Adv}_{\mathrm{loss}}(\mathcal{A})
    &\leq
    \Pr[\mathsf{Bad}]\\
    &\leq
    \varepsilon_{\mathrm{MAC}}(\lambda).
\end{align}
Thus, LoRo-Mark is
$(\varepsilon_{\mathrm{MAC}},t,q,\lambda)$-lossless.
\end{proof}

\section{Experiments}
\label{sec:5}
We evaluate LoRo-Mark around the two central properties established in our design:
\emph{losslessness} and \emph{robustness}.
For losslessness, we examine whether forensic branching introduces any unintended
activation or performance degradation during normal agent execution.
For robustness, we evaluate whether the embedded ownership information remains
recoverable when the forensic behavior sequence is partially modified during
agent repackaging. 

\subsection{Experimental Setup}

\textbf{Dataset.}
We conduct our experiments on \textsc{xLAM-function-calling-60k}~\citep{liu2024apigen},
which contains 60,000 function-calling trajectories covering 3,673 tools
across 21 domains.
We use the natural-language queries as task bodies $T_{\mathrm{body}}$
and their associated tool specifications to construct agent execution tasks.
For utility evaluation, we select three representative domains,
\emph{Data}, \emph{Business}, and \emph{Social}, and sample 100 tasks from
each domain.

\textbf{Agent Setup.}
Unless otherwise specified, we use Qwen3-4B~\citep{yang2025qwen3} as the
default LLM backbone and ReAct~\citep{yao2023react} as the agent framework,
with tools executed through a function-calling interface.
To evaluate whether LoRo-Mark depends on a particular backbone, we additionally
conduct experiments with Qwen3-8B~\citep{yang2025qwen3} and Ministral-8B-Instruct-2410~\citep{mistralai2024ministral}. The experimental results remain highly consistent across all models because both Forensic Detection and Behavioral Embedding operate within the agent harness before the underlying LLM is invoked, making LoRo-Mark independent of the backbone model.  Due to space constraints, the following section focuses on agents built on a single LLM backbone as a representative example. 
For agent utility evaluation, we use DeepSeek-V3~\citep{deepseekai2024deepseekv3} as
the judge model to assess both final responses and execution trajectories.

\textbf{Watermark Configuration.}
LoRo-Mark uses HMAC-SHA256 as the authentication primitive~\citep{bellare1996keying,nist2008hmac}.
Forensic requests are generated by computing an authentication tag over
$T_{\mathrm{body}}$ and encoding the tag into semantically valid parameter
slots, while normal tasks are left unchanged.
Unless otherwise specified, we use the copyright message corresponding to 104 bits, and an interleaved repetition redundancy of $r=3$.

\textbf{Evaluation Metrics.}
We evaluate LoRo-Mark from two primary perspectives.
For watermark verification, we report the \emph{true-positive rate} (TPR)
on authorized forensic requests and the \emph{false-positive rate} (FPR)
on normal or unauthorized requests.
For agent utility, we report Pass Rate (PR), Response Score (RS),
Trajectory Score (TS), and Success Rate (SR)~\citep{huang2026agentmark,wang2026agentwm}, computed over \(N\)
evaluation tasks as
\[
\mathrm{PR}
=
\frac{1}{N}
\sum_{i=1}^{N}
\mathbb{I}
\left[
\text{response}_i \text{ fulfills the request}
\right],
\]
\[
\mathrm{RS}
=
\frac{1}{N}
\sum_{i=1}^{N}
\frac{r_i-1}{4},
\qquad
\mathrm{TS}
=
\frac{1}{N}
\sum_{i=1}^{N}
\frac{t_i-1}{4},
\qquad
\mathrm{SR}
=
\frac{1}{N}
\sum_{i=1}^{N}
\mathbb{I}
\left[
\text{task}_i \text{ is successfully completed}
\right],
\]
where \(r_i,t_i\in\{1,\ldots,5\}\) are the response-quality and
trajectory-quality scores, respectively, and \(\mathbb{I}[\cdot]\)
is the indicator function. Higher values indicate better task performance.
For robustness, we report TPR, computed using the same definition above,
after applying different degrees of behavior substitution or sequence
truncation to the forensic behavior sequence.

\begin{table*}[t]
\centering
\caption{Utility comparison of different agent watermarking schemes.}
\label{tab:aligned-utility}

\renewcommand{\arraystretch}{1.18}
\setlength{\tabcolsep}{4.5pt}

\resizebox{\textwidth}{!}{%
\begin{tabular}{c|cccc|cccc|cccc}
\toprule[1.5pt]

\rowcolor{gray!20}
& \multicolumn{4}{c|}{\textbf{Data}}
& \multicolumn{4}{c|}{\textbf{Business}}
& \multicolumn{4}{c}{\textbf{Social}} \\

\rowcolor{gray!20}
\multirow{-2}{*}{\textbf{Method}}
& \textbf{PR} $\uparrow$
& \textbf{RS} $\uparrow$
& \textbf{TS} $\uparrow$
& \textbf{SR} $\uparrow$
& \textbf{PR} $\uparrow$
& \textbf{RS} $\uparrow$
& \textbf{TS} $\uparrow$
& \textbf{SR} $\uparrow$
& \textbf{PR} $\uparrow$
& \textbf{RS} $\uparrow$
& \textbf{TS} $\uparrow$
& \textbf{SR} $\uparrow$ \\

\specialrule{0.4pt}{1pt}{1pt}
\specialrule{0.4pt}{0pt}{1pt}

Unwatermarked
& 0.6200 & 0.6175 & 0.8850 & 0.9900
& 0.6000 & 0.6025 & 0.8825 & 0.9900
& 0.7100 & 0.6375 & 0.8825 & 0.9800 \\

AGENTWM
& 0.5500 & 0.5700 & 0.8375 & 0.8800
& 0.5100 & 0.5400 & 0.8075 & 0.8500
& 0.6300 & 0.5875 & 0.8150 & 0.8900 \\

\rowcolor{gray!10}
Per. Drop ($\Delta$)
& $-11.29\%$ & $-7.69\%$ & $-5.37\%$ & $-11.11\%$
& $-15.00\%$ & $-10.37\%$ & $-8.50\%$ & $-14.14\%$
& $-11.27\%$ & $-7.84\%$ & $-7.65\%$ & $-9.18\%$ \\

AgentMark
& 0.5600 & 0.5975 & 0.8675 & 0.9600
& 0.5000 & 0.5550 & 0.8300 & 0.9300
& 0.5900 & 0.5900 & 0.8250 & 0.9200 \\

\rowcolor{gray!10}
Per. Drop ($\Delta$)
& $-9.68\%$ & $-3.24\%$ & $-3.03\%$ & $-3.03\%$
& $-16.67\%$ & $-7.88\%$ & $-5.95\%$ & $-6.06\%$
& $-16.90\%$ & $-7.45\%$ & $-6.52\%$ & $-6.12\%$ \\

\specialrule{0.4pt}{1pt}{0pt}

\rowcolor{cyan!10}
\textbf{LoRo-Mark}
& \textbf{0.6200} & \textbf{0.6175}
& \textbf{0.8850} & \textbf{0.9900}
& \textbf{0.6000} & \textbf{0.6025}
& \textbf{0.8825} & \textbf{0.9900}
& \textbf{0.7100} & \textbf{0.6375}
& \textbf{0.8825} & \textbf{0.9800} \\

\rowcolor{cyan!22}
\textbf{Per. Drop ($\Delta$)}
& $\mathbf{0.00\%}$ & $\mathbf{0.00\%}$
& $\mathbf{0.00\%}$ & $\mathbf{0.00\%}$
& $\mathbf{0.00\%}$ & $\mathbf{0.00\%}$
& $\mathbf{0.00\%}$ & $\mathbf{0.00\%}$
& $\mathbf{0.00\%}$ & $\mathbf{0.00\%}$
& $\mathbf{0.00\%}$ & $\mathbf{0.00\%}$ \\

\bottomrule[1.5pt]
\end{tabular}%
}

\end{table*}

\subsection{Losslessness}

We first evaluate whether LoRo-Mark preserves normal agent execution in
practice. According to our theoretical analysis, losslessness relies on two
conditions: ordinary requests should not activate the forensic branch, and,
when the branch is not activated, the agent should follow the original
execution path without watermark intervention. We therefore evaluate both
false activation and task utility.

\textbf{False Activation.}
We first measure whether normal requests accidentally activate the forensic
branch. We randomly sample 10,000 normal tasks from
\textsc{xLAM-function-calling-60k}~\citep{liu2024apigen} and pass them through Forensic Detection.
None of the requests is accepted as a forensic task, yielding $\mathrm{FPR} = \frac{0}{10{,}000}=0.00\%$.

This empirical result is consistent with our theoretical analysis that
unauthorized forensic activation is bounded by the security of the underlying
authentication mechanism.

\textbf{Utility Preservation.}
We next examine whether introducing LoRo-Mark changes the agent's performance
on ordinary tasks. We compare LoRo-Mark with the unwatermarked agent and two
existing agent watermarking methods, AGENTWM~\citep{wang2026agentwm} and
AgentMark~\citep{huang2026agentmark}, on the \emph{Data}, \emph{Business},
and \emph{Social} domains. For each domain, we evaluate 100 tasks and report
Pass Rate (PR), Response Score (RS), Trajectory Score (TS), and Success Rate
(SR), with DeepSeek-V3~\citep{deepseekai2024deepseekv3} serving as the judge. As shown in Table~\ref{tab:aligned-utility}, existing watermarking methods introduce
different degrees of performance degradation because watermark embedding
intervenes in normal behavior selection or execution. In contrast, LoRo-Mark
achieves exactly the same scores as the unwatermarked agent across all three
domains and all four metrics, resulting in $\Delta \mathrm{PR}=
    \Delta \mathrm{RS}
    =
    \Delta \mathrm{TS}
    =
    \Delta \mathrm{SR}
    =
    0$. These results empirically confirm that ordinary requests remain on the
untouched execution path of the original agent, consistent with the
losslessness guarantee established in Section~\ref{sec:losslessness}.

\subsection{Robustness}

We evaluate the robustness of LoRo-Mark against two representative
repackaging operations: partial behavior substitution and sequence truncation.

\textbf{Behavior Substitution.}
We randomly replace a fraction $\rho$ of behaviors in the forensic sequence
and evaluate whether the copyright message can still be recovered.
As shown in Figure~\ref{fig:robustness}(a), LoRo-Mark consistently achieves
higher TPR than AGENTWM~\citep{wang2026agentwm} and AgentMark~\citep{huang2026agentmark} under increasing substitution ratios.
With $r=3$, LoRo-Mark remains reliable under moderate behavior modification.

\textbf{Sequence Truncation.}
We further truncate the forensic sequence and retain only a prefix of the
original execution trace. As shown in Figure~\ref{fig:robustness}(b), with $r=3$, LoRo-Mark achieves 100\% TPR once approximately
34.6\% of the original steps are retained, while shorter prefixes are insufficient for complete recovery. These results confirm that redundant behavioral watermarking preserves ownership evidence under partial modification and truncation of the execution sequence.

\begin{figure}[h]
  \centering
  \includegraphics[width=0.9\linewidth]{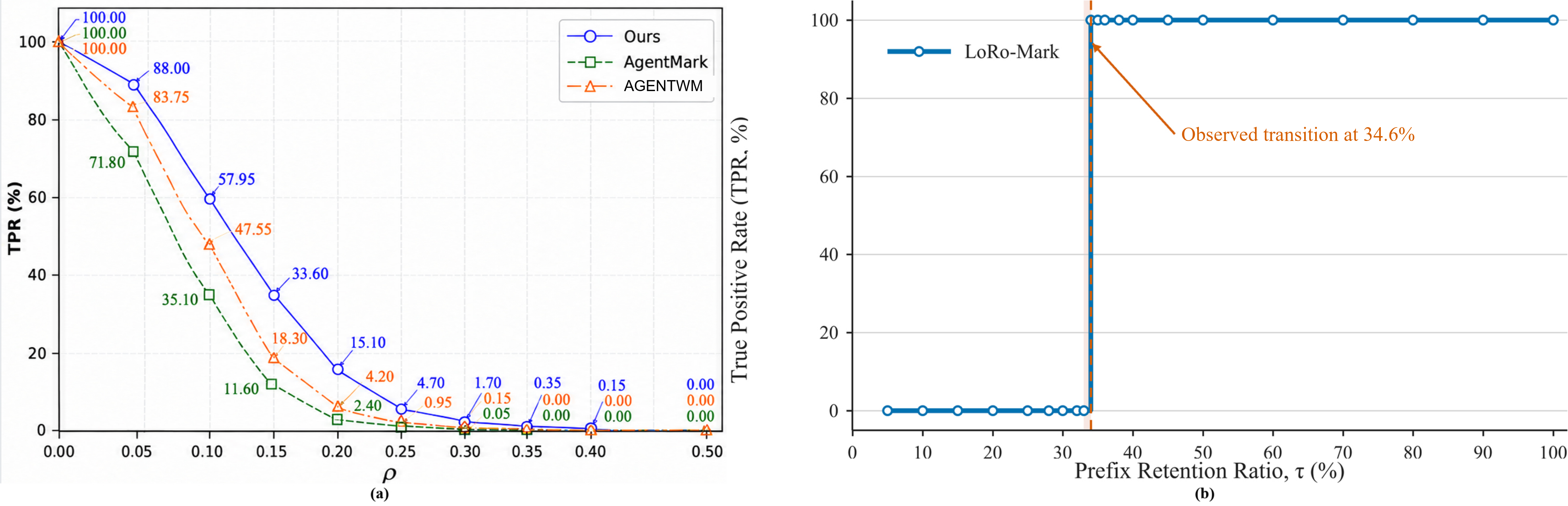}
  \caption{(a) Comparison of TPR for different watermarking schemes under behavior substitution with varying ratios \(\rho\). (b) TPR for LoRo-Mark under different prefix retention ratios.}
  \label{fig:robustness}
\end{figure}
\subsection{Ablation and Parameter Analysis}
\textbf{Effect of Redundancy.}
We first study the effect of the repetition factor $r$, which directly
controls the trade-off between robustness and forensic trace length.
For a copyright message of $|c|$ bytes and a candidate behavior set of
size $|\mathcal{B}|$, LoRo-Mark embeds
$\lfloor \log_2 |\mathcal{B}| \rfloor$ bits per step.
Accordingly, the number of steps required to encode the redundant
copyright message is
$
    S =
    \left\lceil
    \frac{8|c|r}
    {\lfloor \log_2 |\mathcal{B}| \rfloor}
    \right\rceil.
$
Thus, increasing $r$ introduces additional forensic steps but provides
more redundant copies for watermark recovery. To quantify this trade-off, we fix $|\mathcal{B}|=64$ and
$|c|=13$ bytes, and vary $r\in\{3,5,7\}$.
As shown in Table~\ref{tab:tpr-repetition-factor}, larger $r$ consistently
improves TPR under the same behavior-substitution ratio, confirming that
additional redundancy strengthens resistance to sequence corruption~\citep{lin2004error}.
Meanwhile, the required forensic trace length grows linearly with $r$.
We therefore use $r=3$ as the default setting, which provides a practical
balance between robustness and verification overhead.

\textbf{Effect of Parameter Capacity.}
We evaluate how parameter capacity and MAC length jointly affect false activation.
Across 12 configurations with
$\mathrm{capacity\_bits}\in\{14,64,96\}$ and
$\mathrm{mac\_bits}\in\{32,64,96,128\}$,
the empirical FPR consistently follows the theoretical bound
$2^{-\min(\mathrm{capacity\_bits},\mathrm{mac\_bits})}$.
For the 14-bit capacity setting, we observe an FPR of
$5.0\times10^{-5}$, close to the theoretical $2^{-14}$,
while all configurations with at least 32 effective bits produce zero
false activations in 10,000 trials.
These results confirm that the effective security is determined by the
smaller of parameter capacity and MAC length.

\begin{table*}[t]
\centering
\caption{True-positive rates under different behavior-substitution ratios ($\rho$) and repetition factors ($r$).}
\label{tab:tpr-repetition-factor}
\renewcommand{\arraystretch}{1.35}
\setlength{\tabcolsep}{4.5pt}
\resizebox{\textwidth}{!}{%
\begin{tabular}{
>{\centering\arraybackslash}c|
>{\centering\arraybackslash}c|
*{8}{>{\centering\arraybackslash}c|}
>{\centering\arraybackslash}c
}
\toprule[1.5pt]
\rowcolor{gray!20}
&
&
\multicolumn{9}{c}{\textbf{TPR (\%)}}
\\
\rowcolor{gray!20}
\multirow{-2}{*}{\textbf{$r$}}
&
\multirow{-2}{*}{\textbf{Steps}}
&
$\rho=0.05$
&
$\rho=0.10$
&
$\rho=0.15$
&
$\rho=0.20$
&
$\rho=0.25$
&
$\rho=0.30$
&
$\rho=0.35$
&
$\rho=0.40$
&
$\rho=0.50$
\\
\specialrule{0.4pt}{1pt}{1pt}
\specialrule{0.4pt}{0pt}{1pt}

\textbf{3}
&
\textbf{52}
&
100.00
&
88.00
&
57.95
&
33.60
&
15.10
&
4.70
&
1.70
&
0.15
&
0.00
\\

\rowcolor{gray!10}
\textbf{5}
&
\textbf{87}
&
100.00
&
98.65
&
91.10
&
73.95
&
52.20
&
27.45
&
12.30
&
1.30
&
0.00
\\

\textbf{7}
&
\textbf{122}
&
100.00
&
100.00
&
98.70
&
91.25
&
77.65
&
59.00
&
37.00
&
8.05
&
0.60
\\
\bottomrule[1.5pt]
\end{tabular}%
}
\end{table*}

\section{Conclusion}

We study black-box ownership verification under agent repackaging and propose \textbf{LoRo-Mark}, a provably lossless and robust agent watermarking mechanism. LoRo-Mark isolates watermarking from normal execution through an authenticated forensic branch and redundantly distributes ownership information across forensic behavior sequences. We formally establish performance preservation under standard MAC security and empirically demonstrate unchanged utility on normal tasks together with robust watermark recovery under behavior substitution and sequence truncation. Our results suggest a simple design principle for agent watermarking: ownership verification can be separated from normal execution by design.
\subsection*{AI use statement}

In this work, we used generative AI tools to assist with literature research on the agent-repackaging threat scenario and related work. We also used generative AI tools to help refine the experimental design for evaluating losslessness, robustness, and ablation settings; organize relevant benchmarks; and structure and summarize experimental data. After completing the manuscript, we used generative AI tools for language polishing, including correcting grammatical and typographical errors, checking presentation details, and improving clarity and fluency. All AI-assisted research, experimental designs, data summaries, citations, and textual revisions were independently reviewed and verified by the authors. We take full responsibility for the final content of this work, including all text, claims, experimental results, and artifacts produced with the aid of generative AI.

\subsection*{Reproducibility statement}

We provide the information and materials necessary to reproduce our theoretical and empirical results. The main paper describes the threat model, watermark construction, experimental settings, evaluation metrics, baselines, benchmarks, and parameter configurations. The appendix provides detailed algorithms, assumptions, and complete proofs of the losslessness, unforgeability, and robustness results. The supplementary materials include the implementation code, environment dependencies, configuration files, and instructions for reproducing the losslessness, robustness, and ablation experiments. The experimental settings and reported parameters are documented in the accompanying README to facilitate consistent reproduction of the results.

\bibliographystyle{iclr2027_conference}
\bibliography{references}

\clearpage
\appendix
\section{Detailed Algorithms of LoRo-Mark}
\label{sec:supp-algorithms}

This appendix uses the component names and notation introduced in
Section~3.2. In particular, $|c|$ denotes the byte length of the copyright
message, $n=8|c|$ denotes its bit length, and
$c=(c_0,\ldots,c_{n-1})$ denotes its binary representation. When the
available parameter capacity supports only an $\ell$-bit tag, we write
\[
  \operatorname{MAC}_{\ell}(k,x)
  :=\operatorname{Trunc}_{\ell}(\operatorname{MAC}(k,x));
\]
without truncation, $\operatorname{MAC}_{\ell}$ is simply
$\operatorname{MAC}$. The parameter encoder $\mathcal{E}$ and decoder
$\mathcal{D}$ use the same key $k_{\mathrm{emb}}$ and tag length $\ell$.
Their correctness requires
\[
  \mathcal{D}\bigl(\mathcal{E}(T_{\mathrm{body}},\sigma,
  k_{\mathrm{emb}}),k_{\mathrm{emb}}\bigr)=\sigma
\]
for every encodable $\ell$-bit tag $\sigma$.
At each step, $o_t$ denotes the reproducible, domain-separated execution
context, including the step identifier and the observation needed to reconstruct
$\mathcal{B}_t$. We write $\operatorname{WatermarkVerify}$ for
Algorithm~\ref{alg:verify}.

\begin{algorithm}[H]
  \caption{Key Generation}
  \label{alg:keygen}
  \begin{algorithmic}[1]
    \Require Security parameter $\lambda$
    \Ensure $k_{\det}, k_{\mathrm{prv}}, k_{\mathrm{emb}}$
    \State $k_{\mathrm{master}} \gets \{0,1\}^{\lambda}$
    \State $k_{\det} \gets \operatorname{PRF}(k_{\mathrm{master}},\texttt{detect})$~\citep{goldreich1986construct}
    \State $k_{\mathrm{prv}} \gets \operatorname{PRF}(k_{\mathrm{master}},\texttt{prove})$
    \State $k_{\mathrm{emb}} \gets \operatorname{PRF}(k_{\mathrm{master}},\texttt{embed})$
    \State \Return $(k_{\det},k_{\mathrm{prv}},k_{\mathrm{emb}})$
  \end{algorithmic}
\end{algorithm}

\begin{algorithm}[H]
  \caption{Trigger Task Generation}
  \label{alg:trigger-task-gen}
  \begin{algorithmic}[1]
    \Require Valid task body $T_{\mathrm{body}}$, keys $k_{\det},k_{\mathrm{emb}}$, tag length $\ell$
    \Ensure Authenticated forensic request $T^{\ast}$
    \State $\sigma \gets \operatorname{MAC}_{\ell}(k_{\det},T_{\mathrm{body}})$
    \State $P \gets \mathcal{E}(T_{\mathrm{body}},\sigma,k_{\mathrm{emb}})$
    \State $T^{\ast} \gets T_{\mathrm{body}}\oplus P$
    \State \Return $T^{\ast}$
  \end{algorithmic}
\end{algorithm}

\begin{algorithm}[H]
  \caption{Forensic Detection}
  \label{alg:detect-task}
  \begin{algorithmic}[1]
    \Require Incoming request $T=(T_{\mathrm{body}},P)$, keys $k_{\det},k_{\mathrm{emb}}$, tag length $\ell$
    \Ensure Detection result $(s,\sigma)$, where $s\in\{0,1\}$
    \State $\sigma_{\mathrm{expected}} \gets
      \operatorname{MAC}_{\ell}(k_{\det},T_{\mathrm{body}})$
    \State $\sigma_{\mathrm{decoded}} \gets \mathcal{D}(P,k_{\mathrm{emb}})$
    \If{$\sigma_{\mathrm{decoded}}=\sigma_{\mathrm{expected}}$}
      \State \Return $(1,\sigma_{\mathrm{expected}})$
    \Else
      \State \Return $(0,\bot)$
    \EndIf
  \end{algorithmic}
\end{algorithm}

\begin{algorithm}[H]
  \caption{Redundant Encoding and Behavioral Embedding}
  \label{alg:forensic-behavior}
  \begin{algorithmic}[1]
    \Require Copyright message $c$, repetition factor $r$, tag $\sigma$, key $k_{\mathrm{prv}}$
    \Require Reproducible execution contexts $(o_t,\mathcal{B}_t)$ with $|\mathcal{B}_t|\geq2$
    \Require Sufficient cumulative capacity to encode $nr$ bits
    \Ensure Forensic behavior sequence $\tau=(b_1,\ldots,b_S)$
    \State $n\gets8|c|$; $M\gets(c_0\cdots c_{n-1})\|\cdots\|(c_0\cdots c_{n-1})$ with $r$ copies
    \State $k_{\mathrm{step}}\gets\operatorname{MAC}(k_{\mathrm{prv}},\sigma)$
    \State $\tau\gets[\,]$; $p\gets0$; $t\gets1$
    \While{$p<nr$}
      \State $q_t\gets\lfloor\log_2|\mathcal{B}_t|\rfloor$
      \State $\pi_t\gets\operatorname{PRF}(k_{\mathrm{step}},o_t,\mathcal{B}_t)$
      \State $\ell_t\gets\min\{q_t,nr-p\}$
      \State $u_t\gets M[p:p+\ell_t-1]\|0^{q_t-\ell_t}$
      \State $j_t\gets\operatorname{int}(u_t)$; $b_t\gets\pi_t[j_t]$
      \State Append $b_t$ to $\tau$; $p\gets p+\ell_t$; $t\gets t+1$
    \EndWhile
    \State $S\gets|\tau|$
    \State \Return $\tau$
  \end{algorithmic}
\end{algorithm}

\begin{algorithm}[H]
  \caption{Watermark Verification}
  \label{alg:verify}
  \begin{algorithmic}[1]
    \Require $T^{\ast}=(T_{\mathrm{body}},P)$, observed behavior sequence
      $\tau'=(b'_1,\ldots,b'_{S'})$
    \Require Claimed message $c_{\mathrm{claimed}}$, keys
      $k_{\det},k_{\mathrm{prv}},k_{\mathrm{emb}}$, repetition factor $r$, tag length $\ell$
    \Ensure $1$ if ownership verification succeeds; otherwise $0$
    \State $\sigma_{\mathrm{expected}}\gets
      \operatorname{MAC}_{\ell}(k_{\det},T_{\mathrm{body}})$
    \State $\sigma_{\mathrm{decoded}}\gets\mathcal{D}(P,k_{\mathrm{emb}})$
    \If{$\sigma_{\mathrm{decoded}}\neq\sigma_{\mathrm{expected}}$}
      \State \Return $0$
    \EndIf
    \State $n\gets8|c_{\mathrm{claimed}}|$;
      $k_{\mathrm{step}}\gets\operatorname{MAC}(k_{\mathrm{prv}},\sigma_{\mathrm{expected}})$
    \State $\widehat{M}\gets[\,]$; $p\gets0$
    \For{$t\in\{1,\ldots,S'\}$}
      \State Reconstruct $(o_t,\mathcal{B}_t)$ from $T^{\ast}$ and the exposed execution context
      \State $q_t\gets\lfloor\log_2|\mathcal{B}_t|\rfloor$
      \State $\pi_t\gets\operatorname{PRF}(k_{\mathrm{step}},o_t,\mathcal{B}_t)$
      \State $j'_t\gets\operatorname{IndexOf}(b'_t,\pi_t)$
        \Comment{$\bot$ if $b'_t\notin\mathcal{B}_t$}
      \State $\ell_t\gets\min\{q_t,nr-p\}$
      \If{$j'_t\neq\bot$ \textbf{and} $j'_t<2^{q_t}$}
        \State Append the first $\ell_t$ bits of $\operatorname{bits}(j'_t,q_t)$ to $\widehat{M}$
      \Else
        \State Append $\ell_t$ erasure symbols to $\widehat{M}$
      \EndIf
      \State $p\gets p+\ell_t$
      \If{$p=nr$}
        \State \textbf{break}
      \EndIf
    \EndFor
    \If{$p<nr$}
      \State Append $nr-p$ erasure symbols to $\widehat{M}$
    \EndIf
    \State $\hat{c}\gets\operatorname{MajorityDecode}(\widehat{M},r,n)$~\citep{lin2004error}
      \Comment{Ignore erasures; reject empty groups or ties}
    \State \Return $\mathbb{I}[\hat{c}=c_{\mathrm{claimed}}]$
  \end{algorithmic}
\end{algorithm}

\section{Supplementary Proofs}
\label{sec:supp-proofs}

\subsection{Correctness (No False Negatives)}

\begin{proposition}[No false negatives]
\label{prop:nfn}
Let $c$ be the copyright message configured in the protected agent, let $r$
be a positive odd repetition factor, and let $T^{\ast}$ be produced by
Trigger Task Generation using a correct parameter encoder--decoder pair. If
$\tau$ is the intact forensic behavior sequence
produced by Redundant Encoding and Behavioral Embedding under reconstructible
execution contexts, then
\[
  \operatorname{WatermarkVerify}
  (T^{\ast},\tau,c,k_{\det},k_{\mathrm{prv}},k_{\mathrm{emb}},r,\ell)=1.
\]
\end{proposition}

\begin{proof}
Trigger Task Generation constructs $T^{\ast}=(T_{\mathrm{body}},P)$ with
\[
  \mathcal{D}(P,k_{\mathrm{emb}})
  =\operatorname{MAC}_{\ell}(k_{\det},T_{\mathrm{body}})=\sigma.
\]
Consequently, Forensic Detection returns $(1,\sigma)$ and activates the
forensic branch. Both Behavioral Embedding and Watermark Verification derive
$k_{\mathrm{step}}=\operatorname{MAC}(k_{\mathrm{prv}},\sigma)$ and reconstruct
the same $\pi_t=\operatorname{PRF}(k_{\mathrm{step}},o_t,\mathcal{B}_t)$ at
every step. Therefore, each intact behavior $b_t=\pi_t[j_t]$ is mapped back
to the same index $j_t$ and hence to the same $q_t$-bit chunk. Because the
verifier retains only the first $\ell_t$ bits of the final chunk, their
concatenation is exactly
\[
  M=(c_0\cdots c_{n-1})\|\cdots\|(c_0\cdots c_{n-1}),
  \qquad n=8|c|,
\]
with $r$ identical copies. Majority decoding thus returns $\hat c=c$, so
Watermark Verification returns $1$.
\end{proof}

\subsection{Behavior-Sequence Unforgeability}

The following one-shot statement excludes replay of an authentic forensic
sequence generated for the same execution contexts. Such a freshness
condition is necessary because any accepted sequence could otherwise be
replayed verbatim.

\begin{proposition}[Fresh-context behavior-sequence unforgeability]
\label{prop:trajectory-unforgeability}
Let $n=8|c|$ and let $r$ be positive and odd. Consider a one-shot PPT adversary that
knows $T^{\ast}$, $k_{\det}$, and $k_{\mathrm{emb}}$, but not
$k_{\mathrm{prv}}$, has no oracle access to the target per-step permutations,
and has not observed an authentic forensic behavior for any of the pairwise
distinct, domain-separated target execution contexts. Assume also that
$n\geq\max_t q_t$. Suppose the derivation of
$k_{\mathrm{step}}$ is $\varepsilon_{\mathrm{KDF}}(\lambda)$-pseudorandom and
the domain-separated PRF evaluations are jointly indistinguishable from
independent uniform per-step permutations with multi-context advantage
$\varepsilon_{\mathrm{PRF}}(\lambda,S')$. Then
\begin{align}
  &\Pr\!\left[
    \operatorname{WatermarkVerify}
    (T^{\ast},\tau',c,k_{\det},k_{\mathrm{prv}},k_{\mathrm{emb}},r,\ell)=1
  \right] \\
  &\qquad\leq 2^{-n}
    +\varepsilon_{\mathrm{KDF}}(\lambda)
    +\varepsilon_{\mathrm{PRF}}(\lambda,S').
\end{align}
\end{proposition}

\begin{proof}
Replace the derivation of $k_{\mathrm{step}}$ by a uniform key and then replace
the domain-separated PRF evaluations by independent uniform permutations.
The two hybrid transitions incur advantages at most
$\varepsilon_{\mathrm{KDF}}(\lambda)$ and
$\varepsilon_{\mathrm{PRF}}(\lambda,S')$, respectively.

In the ideal experiment, a supplied nonmember behavior is an erasure. For
each supplied $b'_t\in\mathcal{B}_t$, its index under the independent uniform
permutation is uniform on $\{0,\ldots,|\mathcal{B}_t|-1\}$. Condition on the
resulting per-step active-or-erasure pattern. At every active step,
conditioning on $j'_t<2^{q_t}$ leaves $j'_t$ uniform on
$\{0,\ldots,2^{q_t}-1\}$, so the recovered payload positions at that step are
independent unbiased bits. Independence of the ideal per-step permutations
makes these bits independent across active steps.

Since $n\geq\max_t q_t$, the $r$ positions associated with any fixed
copyright bit lie in distinct steps. For each copyright bit, an empty group
or a tie is rejected; otherwise bit-complement symmetry implies that its
surviving majority equals the fixed target bit with probability at most
$1/2$. Distinct copyright bits use disjoint random bit positions, so the
conditional acceptance probability is at most $2^{-n}$. Averaging over the
active-or-erasure pattern and adding the two hybrid losses proves the claim.
\end{proof}

\subsection{Robustness}
\label{sec:supp-robustness}

For closed-form bounds matching the experiments, suppose
$|\mathcal{B}_t|=2^{\beta}$ at every embedding step, so each step carries the
constant capacity $\beta=\log_2|\mathcal{B}_t|\geq1$. Let $n=8|c|$, let $r$
be positive and odd, and let
\[
  M=(c_0\cdots c_{n-1})\|\cdots\|(c_0\cdots c_{n-1})
\]
contain $r$ copies and have length $N=nr$. Copy $k$ of bit $c_i$ occupies
position $i+kn$ and step
\[
  t_{i,k}=\left\lfloor\frac{i+kn}{\beta}\right\rfloor+1,
  \qquad
  S=\left\lceil\frac{nr}{\beta}\right\rceil.
\]
Watermark Verification ignores erased copies and takes a majority over the
surviving copies, rejecting when none survives or when a tie occurs.

\begin{proposition}[Robustness to prefix truncation]
\label{prop:robust-truncation}
Let $\alpha\in[0,1]$. Suppose only the first
$S'=\lfloor\alpha S\rfloor$ steps of an otherwise
uncorrupted forensic behavior sequence are retained. If
\begin{equation}
  \alpha\geq
  \frac{\lceil n/\beta\rceil}{S}
  \approx\frac{1}{r},
\end{equation}
then Watermark Verification returns $1$ with probability $1$.
\end{proposition}

\begin{proof}
The retained prefix contains the first $S'\beta$ encoded bits. All positions
$0,\ldots,n-1$, which form the first complete copy of $c$, are retained when
$S'\beta\geq n$, equivalently when
$S'\geq\lceil n/\beta\rceil$. The assumed lower bound on $\alpha$ guarantees
this condition. Every copyright bit therefore has at least one surviving
copy, and all surviving copies are correct. Erasure-aware majority decoding
returns $\hat c=c$, so verification succeeds.
\end{proof}

\begin{remark}
For the default configuration $n=104$, $\beta=6$, and $r=3$, we have $S=52$
and the sufficient retained-prefix threshold is
$\lceil104/6\rceil/52=18/52\approx34.6\%$.
\end{remark}

\begin{proposition}[Robustness to random behavior substitution]
\label{prop:robust-substitution}
Let $\rho\in[0,1]$ and assume $n\geq\beta$. Independently at each step,
replace the authentic
behavior with probability $\rho$ by a behavior sampled uniformly from
$\mathcal{B}_t$. Then
\begin{equation}
  \Pr\!\left[
    \operatorname{WatermarkVerify}
    (T^{\ast},\tau',c,k_{\det},k_{\mathrm{prv}},k_{\mathrm{emb}},r,\ell)=1
  \right]
  \geq 1-np_{\mathrm{bit}}(\rho),
\end{equation}
where, for $X\sim\operatorname{Bin}(r,\rho/2)$,
\begin{equation}
  p_{\mathrm{bit}}(\rho)
  :=\Pr\!\left[X\geq\left\lceil\frac{r}{2}\right\rceil\right].
\end{equation}
\end{proposition}

\begin{proof}
The $r$ copies of $c_i$ occur at positions
$i,i+n,\ldots,i+(r-1)n$. Since $n\geq\beta$, these positions lie in distinct
steps. Their substitution events and replacement draws are therefore
independent under the stated attack model.

An unsubstituted copy is correct. Conditional on substitution, the
replacement behavior is uniform in $\mathcal{B}_t$ and hence its index under
the fixed permutation $\pi_t$ is uniform in
$\{0,\ldots,2^{\beta}-1\}$. Each decoded bit is then wrong with probability
$1/2$, giving the unconditional per-copy error probability $\rho/2$.
Consequently, the number of wrong copies of $c_i$ follows
$\operatorname{Bin}(r,\rho/2)$, and majority decoding fails for this bit with
probability $p_{\mathrm{bit}}(\rho)$. A union bound over the $n$ copyright
bits gives a total failure probability at most
$np_{\mathrm{bit}}(\rho)$, proving the result.
\end{proof}

\begin{remark}
For $r=3$,
\[
  p_{\mathrm{bit}}(\rho)
  =3\left(\frac{\rho}{2}\right)^2
    \left(1-\frac{\rho}{2}\right)
   +\left(\frac{\rho}{2}\right)^3.
\]
This equals approximately $0.007$, $0.028$, and $0.156$ at
$\rho=0.1$, $0.2$, and $0.5$, respectively. The union bound may be loose
because different copyright bits can share step-level substitution events;
no cross-bit independence is required for the bound.
\end{remark}

\end{document}